\documentclass[12pt]{amsart}
\usepackage{amsmath, amssymb, graphicx, mathrsfs, dsfont, hyperref, multirow}

\usepackage[usenames,dvipsnames,svgnames]{xcolor}
\usepackage[top=1in, bottom=1in, left=1in, right=1in]{geometry} 

\makeatletter
\newcommand*{\rom}[1]{\expandafter\@slowromancap\romannumeral #1@}
\makeatother

\renewcommand{\min}{\operatorname{min}}

\theoremstyle{plain}
\newtheorem{theorem}{Theorem}[section]
\newtheorem{lemma}[theorem]{Lemma}

\theoremstyle{remark}
\newtheorem*{remark}{Remark}

\begin{document}

\title[Brownian motion on the Spider like Quantum graphs]{Brownian motion on the Spider like Quantum graphs}

\author[Stanislav molchanov]{Stanislav Molchanov}

\address{Department of Mathematics and Statistics, University of North Carolina at Charlotte, 9201 University City Blvd., Charlotte, NC 28223, USA}

\email{smolchan@charlotte.edu}

\author[Madhumita Paul]{Madhumita Paul}

\address{Department of Mathematics and Statistics, University of North Carolina at Charlotte, 9201 University City Blvd., Charlotte, NC 28223, USA}

\email{mapaul3@charlotte.edu}


\thanks{2020\textit{Mathematics Subject Classification:}Primary 60J65; secondary 47A11 \\\textit{Keywords and Phrases: Brownian motion, Transition probability, Quantum graphs, Spider Laplacian, Spectrum, Absolute continuous spectrum.}}
\begin{abstract}
The paper contains the probabilistic analysis of the Brownian motion on the simplest quantum graph, spider: a system of $N$-half axis connected only at the graph's origin by the simplest (so-called Kirchhoff's)  gluing conditions. The limit theorems for the diffusion on such a graph, especially if $N \to \infty$ are significantly different from the classical case $N =2$ (full axis). Additional results concern the properties of the spectral measure of the spider Laplacian and the corresponding generalized Fourier transforms. The continuation of the paper will contain the study of the 
spectrum for the class of Schr\"{o}dinger operators on the spider graphs: Laplacian perturbed by unbounded potential and related phase transitions. 
\end{abstract}

\thispagestyle{empty}

\maketitle
\section{Introduction}

The quantum (or metric) graph concept is closely related to the physical, primarily, optical applications. The network of thin cylindrical channels, drilled in the optical materials (like silicon) asymptotically, if $\epsilon$ (the diameter of channels) tends to 0, forms a quantum graph. The Maxwell equation, in this case, degenerates into second order scalar equation with appropriate gluing conditions in the branching points of the network (see details in~\cite{berkolaiko2013introduction, molchanov2007scattering}).\\

 Corresponding Hamiltonian generates the continuous diffusion process on the quantum graph. Locally, (outside the branching point) this process is the usual $1-D$ Brownian motion. The new effects appear only in the neighborhood of the branching point. The graph in this case has the spider structure: the single point, the origin or the head of the spider, and finitely many legs.\\
\subsection{Spider graph}
 The spider graph $Sp_N$ with $N$ legs $l_i$ for $ i = 1, 2,\dots,  N$ is the group of $N$ half-axis, connected at the single point $0$ (origin of the graph). Along each leg we can then introduce the Euclidean coordinates $x_i \geq 0$, $i = 1, 2,\dots, N$ (See figure \ref{fig 1}). The metric $d(x_i,y_i)$ on $Sp_N$ has the form,

\begin{align*}
d\left(x_{i}, y_{i}\right)=|x_{i}-y_{i}|, && \text { for $ x_{i}, y_{i} \in l_{i}$} \\
d\left(x_{i}, y_{j}\right)=|x_{i}|+|y_{j}| && \text { for $ |x_{i}|=|x_{i}-0|$, $x_{i} \in l_{i}$, $y_{j} \in l_{j}$ for $ i \neq j$}
\end{align*}

 The Lebesgue measure on $S p_{N}$ also has an obvious meaning. The simplest functional spaces such as, $\mathbb{C}(S p_{N})$ (space of the bounded continuous functions with the norm $||f||_{\infty}=\sup _{x \in S p_{N}}|f(x)|$ or,
\begin{align*}
\mathbb{L}^{2}(S p_{N})=\left\{f(x): \int_{S p_{N}} f^{2}(x) d x=||f||_{2}^{2}=\sum_{i=1}^{N} \int_{0}^{\infty}(f^{2}(x_{i}) d x_{i})\right\}
\end{align*}

have the usual properties (completeness etc.).
The new definition requires only the space of the smooth functions and the differential operators (first of all, the Laplacian (see [1]). For the compactly supported smooth $(\mathbb{C}^{2}$-class) functions $f(x)$, whose support does not contain the origin, Laplacian is simply the second derivative with respect to local coordinates $x_{i}$ on the legs $l_{i}, i=1,2, \dots, N$.

 \begin{figure}[h]
\centering
{\includegraphics[scale=.8]{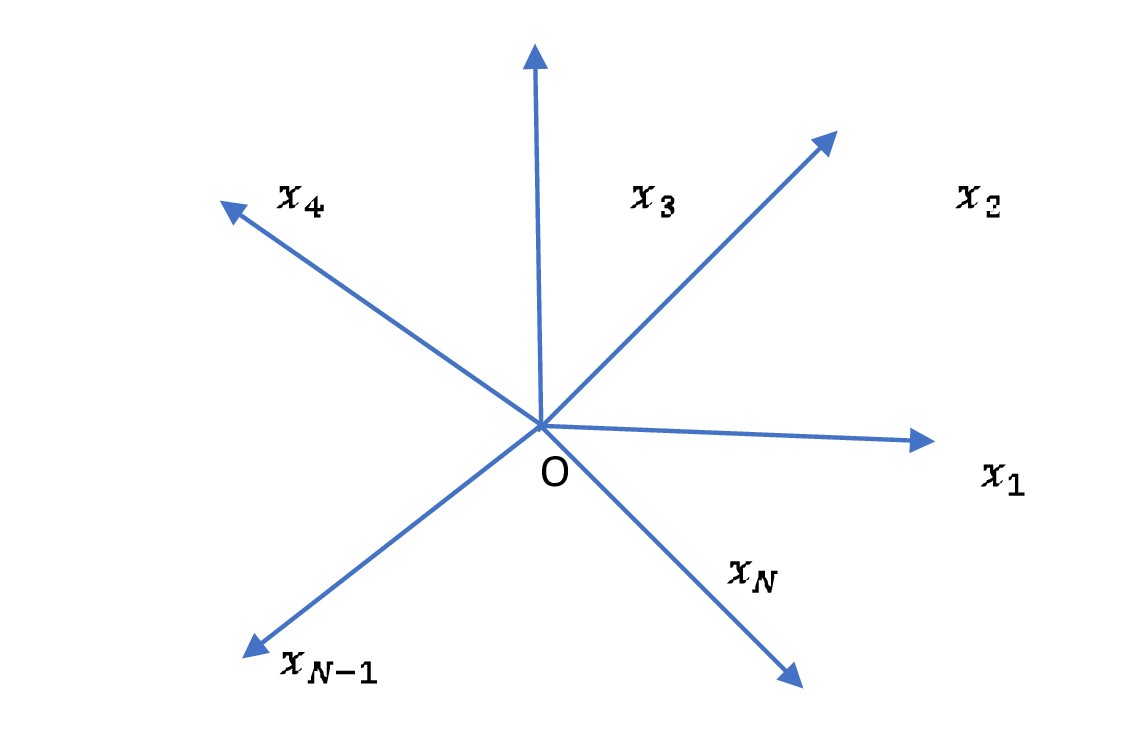}}
\caption{Spider graph with N legs}
 \label{fig 1}
\end{figure}

\begin{figure}[h]
\centering
{\includegraphics[scale=.8]{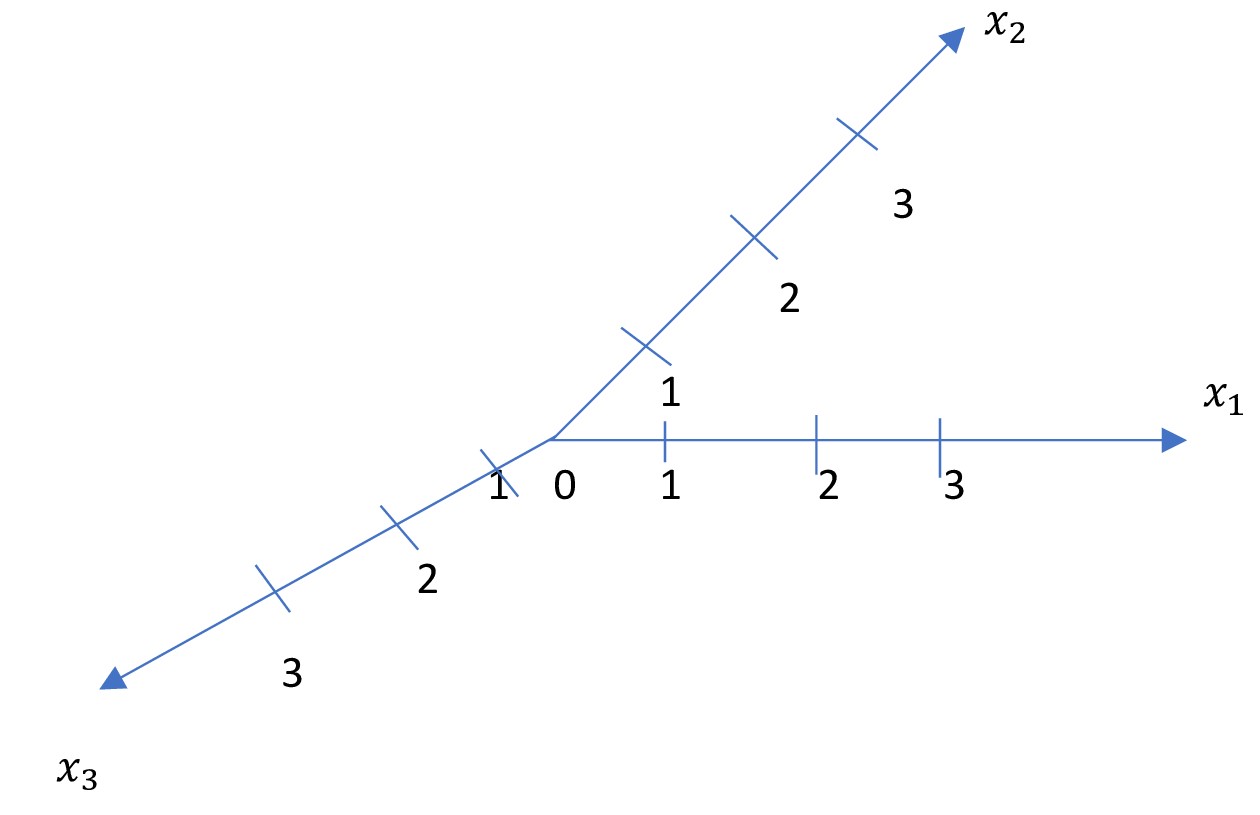}}
\caption{Lattice spider with 3 legs}
 \label{fig 2}
\end{figure}

At the origin (branching point) we will introduce Kirchhoff's gluing conditions for $f$ from the domain of the definition of Laplacian:
\begin{align}
\text{Continuity at $x=0$, i.e $f(x) \in \mathbb{C}(S p_{N})$. In particular $f(0)=\lim _{x_{i} \rightarrow 0} f(x_{i}),  i=1,2 \dots, N$ } \label{a}\\
\text{ For each $i=1,2, \dots, N$, there exists, $\frac{d f}{d x_{i}}(0)=f_{i}^{\prime}$ and $\sum_{i=1}^{N} f_{i}^{\prime}=0$}\label{b}
\end{align}
Now we can define the Sobolev space $H_{2,1}$ on $Sp_{N}$ ($f \in H_{2,1}$ if $\int_{S p_{N}}(f^{\prime})^{2}(x) d x=||f||_{2,1}^{2}$) and try to construct the Brownian motion $b(t), t \geq 0$ with generator $\Delta$ or the Dirichlet quadratic form:
\begin{align*}
\left(f^{\prime}, g^{\prime}\right)=\int_{S p_{N}} f^{\prime}(x) g^{\prime}(x) d x
\end{align*}

\subsection{Random walk}
 Let's try to clarify the meaning of the Brownian motion and Kirchhoff's gluing condition. Roughly, on each leg it is the usual $1-D$ Wiener process until the first exit to the origin. In the corresponding Markov moment $\tau_{0}$, the process is selected with probability $\frac{1}{N}$ on one of the legs and starts to move along this leg.
 The justification of this rough statement can be based on the discretization of space and time. We will start from the standard unit scales on half-lattices $\mathbb{Z}_{+}^{d}$ and time $t=0, \pm 1, \dots$\\

 Consider the random walk $x(t)$ on the lattice spider $\mathbb{Z}\left(S p_{N}\right)$ : (see figure \ref{fig 2}, $N=3$ ) such that for any $x_{i} \neq 0$, $i=1,2, \dots, N$
\begin{align*}
&&Q\{x(n+1)=x_{i} \pm 1 / x(n)=x_{i}\}=\frac{1}{2}\\
and&& Q\{x(n+1)=1_{i} / x(n)=0\}=\frac{1}{N}
\end{align*}
Let us introduce the important random variable (r.v):
\begin{align*}
\tau_{0}=\min \left\{n: x(n)=0 / x(0)=x \in leg_{i}\right\}
\end{align*}
Note that, the random variable $\tau_{0}$ is finite $P$-a.s (due to recurrence of $1-D$ symmetric random walk). The generating function $\psi_{i}(x, z)=E_{x_{i}}z^{\tau_{0}}, i=1,2, \dots, N$ is independent on the number $i$ of the leg and
\begin{align}
\psi_{i}(x, z)=\left(\frac{1-\sqrt{1-z^{2}}}{z}\right)^{x_{i}}
\end{align}
One can prove now that,
\begin{align*}
Ee^{-\lambda \Theta_{a}}= e^{-\sqrt{2\lambda}a} \qquad \text{where $\Theta_{a}=\lim\limits_{n \to \infty} \frac{\tau_0(a)}{n^2}$ (in law)}
\end{align*}
(here $x=[na]$ is the initial point of the random walk)(see~\cite{feller1991introduction}).\\

  Let us also introduce the transition probabilities of the random walk on $\mathbb{Z}\left(S p_{N}\right)$,
\begin{align*}
\text { a) } Q\left(n, 0, x_{i}\right)=Q_{0}\left\{x(t)= x_{i} \in l_{i} / x(0)=0\right\}=\frac{2}{N} \tilde{Q}(n, 0, x)
\end{align*}
where, $\tilde{Q}(n, 0, x)=\left(\begin{array}{l}n \\ |x|\end{array}\right) \left(\frac{1}{2}\right)^{n}$ is the transition probability of the symmetric random walk on $\mathbb{Z}$
\begin{align*}
\text { b) }Q\left(n, x_{i}, y_{j}\right)=\sum_{k} Q_{x_{i}}\left\{\tau_{0}=k\right\} Q\left(n-k, 0, y_{j}\right)
\end{align*}
In this formula $x_{i}, y_{j}$ belong to different legs and $Q_{x_i}(\centerdot)$ denotes the probability measure associated with the random walk $x(t)$, $t \geq 0$, $x(0)=x_i$
\begin{align*}
\text { c) } Q\left(n, x_{i}, y_{i}\right)=Q_{-}\left(n, x_{i}, y_{i}\right)+\sum_{k=1}^{n} Q\left\{\tau_{0}=k\right\} Q\left(n-k, 0, y_{i}\right)
\end{align*}
Here, $x_i$, $y_i$ belong to $l_{i}$ and $Q_{-}\left(n, x_{i}, y_{i}\right)=Q\left(n, x_{i}, y_{i}\right)-Q\left(n, x_{i},-y_{i}\right)$ is the transition probability of the random walk on $\mathbb{Z}_{+}^{1}$ with annihilation at $x=0$. Also, $Q_{+}= Q(n, x_i, -y_i) + Q(n, x_i, y_0)$ where $x_i$, $y_i \in l_i $\\

 \section{Transition probabilities for the Brownian motion on $Sp_N$}

Using the CLT and the limit theorem for the r.v. $\tau_{0}$ and passing to the limit after standard space-time normalization one can find formulas for transition probabilities of the Brownian motion $b(t)$ of $S p_{N}$. It is not difficult to check these formulas (after simplified transformations) directly without the discussion of discritization~\cite{paul2023spectral}.
\begin{align}
p\left(t, 0, y_{j}\right) & =\frac{1}{N} p^{+}\left(t, 0, y_{j}\right)=\frac{2}{N} \frac{e^{-\frac{y^{2}}{2 t}}}{\sqrt{2 \pi t}} \label{4}\\
p\left(t, x_{i}, y_{j}\right) &= 
\int_{0}^{t} \frac{x_{i}}{\sqrt{2 \pi s^{3}}} e^{-\frac{x_{i}^{2}}{2 s}} \times \frac{2}{N} \frac{e^{-\frac{y_{j}^{2}}{2(t-s)}}}{\sqrt{2 \pi(t-s)}} d s &&\text{$x_i$, $y_j$ are on different legs} \label{5}
\end{align}
Similarly,
\begin{align}
\nonumber
p\left(t, x_{i}, y_{i}\right) & =p_{-}\left(t, x_{i}, y_{i}\right)+p_{+}\left(t, x_{i}, y_{i}\right) \\ \nonumber
& =p_{-}\left(t, x_{i}, y_{i}\right)+\int_{0}^{t} \frac{x_{i}}{\sqrt{2 \pi s^{3}}} e^{-\frac{x_{i}^{2}}{2 s}} \times \frac{2}{N} \frac{e^{-\frac{y_{i}^{2}}{2(t-s)}}}{\sqrt{2 \pi(t-s)}} d s \\
& =\frac{1}{\sqrt{2 \pi t}}\left[e^{-\frac{\left(x_{i}-y_{i}\right)^{2}}{2 t}}-e^{-\frac{\left(x_{i}+y_{i}\right)^{2}}{2 t}}\right]+\frac{2}{N} \int_{0}^{t} \frac{x_{i}}{\sqrt{2 \pi s^{3}}} e^{-\frac{x_{i}^{2}}{2 s}} \times \frac{e^{-\frac{y_{i}^{2}}{2(t-s)}}}{\sqrt{2 \pi(t-s)}} d s \label{6}
\end{align}
Here $x_i$ and $y_i$ are two points on the $i^{th}$ leg. $p_{+}$, $p_{-}$ are transition densities for the Brownian motion on $[0, \infty)$ with Dirichlet or Neumann boundary condition.\\
\begin{remark}
Sign $+$ corresponds to the first visit of point 0.
Sign $-$ corresponds to $\{\tau_0 > t \}$
\end{remark}

 Our goal now is to prove several limit theorems describing the space fluctuations of the process $b(t)$ on $S p_{N}$.\\
The formulas \eqref{5}, \eqref{6} can be simplified. Consider spider graph $Sp(2)$ for $N=2$, it is the usual real line $\mathbb{R}$ with point $x=0$ (origin) which divides the real line into two half axis.\\

 Let $x>0$, $-y>0$ (i.e. $y<0$) are two points and on $\mathbb{R}$ and $q(t,x,y) = \frac{e^{-\frac{(x-y)^2}{2t}}}{\sqrt{2 \pi t}}$ is the standard transition density for $1$-D Brownian motion $b(t)$. Then the same calculations as above give (due to the strong Markov property of $b(t)$ at the moment $\tau_0=\min(s: b(s) =0 / b(0)=x $))

\begin{align}
\frac{e^{-\frac{(x+y)^2}{2t}}}{\sqrt{2 \pi t}}=  \int_{0}^{t} \frac{x}{\sqrt{2 \pi s^{3}}} e^{-\frac{x^{2}}{2 s}} \times \frac{e^{-\frac{y^{2}}{2(t-s)}}}{\sqrt{2 \pi(t-s)}} \label{7}
\end{align}

Using \eqref{7} we will get , for $i \neq j$
\begin{align}
p(t,x_i,y_j) = \frac{e^{-\frac{(x_i+y_j)^2}{2t}}}{\sqrt{2 \pi t}}\centerdot \frac{2}{N} &&  \text{(instead of \eqref{5})} \label{7*}
\end{align}
 and (using \eqref{4}) for the $leg_i$
\begin{align}
\nonumber 
p\left(t, x_{i}, y_{i}\right) &= \frac{1}{\sqrt{2 \pi t}}\left[e^{-\frac{\left(x_{i}-y_{i}\right)^{2}}{2 t}}-e^{-\frac{\left(x_{i}+y_{i}\right)^{2}}{2 t}}\right]+\frac{2}{N}  \frac{e^{-\frac{(x_i+y_i)^2}{2t}}}{\sqrt{2 \pi t}}\\
&=\frac{1}{\sqrt{2 \pi t}}\left[e^{-\frac{\left(x_{i}-y_{i}\right)^{2}}{2 t}}-(\frac{N-2}{N})e^{-\frac{\left(x_{i}+y_{i}\right)^{2}}{2 t}}\right] \label{8}
\end{align}
Simple formulas \eqref{7*}, \eqref{8} define the transition probabilities for $b(t)$ on $Sp(N)$, $N \geq 2$. \\

 In the theory of diffusion process we can work either with the semigroup $P_t$, given by transition probabilities: 
\begin{equation*}
(P_t f) (x) = \int_{Sp_N} f(y) p(t,x,y) dy 
\end{equation*}
or, with the trajectories (in our case $b(t)$): 
\begin{equation*}
(P_t f) (x) = E_x f(b(t)) 
\end{equation*}
Now we will define the trajectories of $b(t)$ on the $Sp(N)$. Let $b^{1}(t), \dots, b^{N}(t)$ are independent $1$-D Brownian motions with transition probabilities. 
 
$$ q^{l}(t,x_i,y_i) = \frac{e^{-\frac{(x_i-y_i)^2}{2t}}}{\sqrt{2 \pi t}} $$
 The processes 
$b_{+}^{i}(t)$ are reflected Brownian motions on the half axis $[0, \infty)$ with transition densities (for $x_i, y_i \geq 0$)
\begin{align*}
q_{+}^{i}\left(t, x_{i}, y_{i}\right) &= \frac{1}{\sqrt{2 \pi t}}\left[e^{-\frac{\left(x_{i}-y_{i}\right)^{2}}{2 t}}+e^{-\frac{\left(x_{i}+y_{i}\right)^{2}}{2 t}}\right]
\end{align*}
Consider now the sequence of the small numbers $\delta_{l}= 2^{-l}$, $l=0,\dots $ and define the Markov moment $\tau_{l}^{0}=\min(t:b_{+}^{i}(t)=\delta_l)$, on one of the half-axes $l_i=1,2, \dots, N$. It is easy to see, 
\begin{align*}E_{0}e^{-\lambda \tau_{l}^{0}}= \frac{1}{\cosh \sqrt{2\lambda} \cdot \delta_l}, && E_0 \tau_{l}^{0}=\delta_{l}^{2}\end{align*}
and the random variable $\tau_{l}^{0}$ on the the spider for the Brownian motion $b(t)$, starting from 0 has the same law as the 
$\tau_{i,l}^{0}= \min(t: b_{+}^{i}(t)= \delta_l/ b_{+}^{i}(0)=0)$. Here $b_{+}^{i}(t)$ is the Brownian motion on $l_i=[0,\infty)$ with reflection boundary condition at 0.

 Let us describe now the semi Markov process $x_{l}(t)$ on $Sp(N)$ for fixed small parameters $\delta_{l}=2^{-l}$, for $l=1,2,\dots, $. This process starts from 0 and stays at 0 until moment $\tau_{l,1}$ (this is the only deviation of $x_l(t)$ from the Markov property, since the random variable $\tau_{l,1}$ is not exponentially distributed). At the moment $\tau_{l,1}$ it jumps with probability $\frac{1}{N}$ at one of the points $\delta_{l} \in l_1$ and completes the process $b_{+}^{i}$ until it returns to 0 (corresponding random variable has the Laplace transform $E_{\delta_{l}}e^{-\lambda \Theta_{l,1}}=exp(-\sqrt{2\lambda}\delta_l)$. Then it stays at 0 and after random time $\tau_{l,2}$, jumps at one of the point $\delta_{l}=2^{-l} \in l_{j}$ with probability $\frac{1}{N}$ etc.\\

 The process $x_{l}(t)$ contains the central point: the motion from the origin on one of the legs with probability $\frac{1}{N}$. One can realize the processes $x_l(t)$ on one probability space. Let us outline the simple element of this construction. Consider process $x_l(t)$, $l \geq 2$ and construct $x_{l-1}(t)$ in terms of $x_{l}(t)$. At the moment $\tau_{l,1}$ the process $x_{l}(t)$ jumps from 0 to one of the points $2^{-l}$ on $l_i$, $i=1,2, \dots, N$ after it moves on $[0, 2^{-l+1}]$ along the process $b_{+}^{i}(t)$. The first exit from $[0, 2^{-l+1}]$ can be either at the point $2^{-l+1}$ and we say in this case that, $x_{l-1}(\centerdot)= \delta_{l-1}$. If at this moment $x_{l}(t)=0$ then we stay at 0 at random time with the law $\tau_{l, \centerdot}$ then jump to one of the points $\delta_{l,j}$ with probability $\frac{1}{N}$ and repeat this procedure until we will enter to $\delta_{l-1,j}$. It is easy to check that this random moment has Laplace transform $e^{-\sqrt{2\lambda}\delta_{l-1}}$, $\delta_{l-1}=2^{-(l-1)}$, i.e we realized process $x_{l-1}(t)$ in terms of $x_l(t)$. Passing to the limit we can define the trajectories of the Brownian motion on $Sp_N$. Our goal is to prove several limit theorems for $b(t)$ on the spider graph. 
\subsection{First exit time from the ball $Sp(N,L)$} 
Let $L$ is the large parameter and $Sp(N,L)$ is the spider with the center $0$ and $N$ legs of the length $L$.\\

  Put, 
  \begin{align*}
&&\partial_{L}=\{ x_i =L, i=1,2, \dots, N \} = \text{boundary of $Sp(N,L)$}\\
\text{and}&&\tau_{L}=\min \left(t: b(t) \in \partial_{L}\right)
\end{align*}
\begin{theorem}
If $\tau_{L}=\min \left(t: b(t) \in \partial_{L}\right) $ then, 
\begin{align}&& E_{x} e^{-\lambda \tau_{L}}&=\frac{\cosh \sqrt{2 \lambda} x_{i}}{\cosh \sqrt{2 \lambda} L}, & i=1,2, \dots, N \label{9}\end{align}
\begin{align}\text{in particular,} && E_{0} e^{-\lambda \tau_{L}}&=\frac{1}{\cosh \sqrt{2 \lambda} L}, & i=1,2, \dots, N \label{10}\end{align}
\end{theorem}

\begin{proof}
The last formula \eqref{10} was already used in the construction of the process $x_{l}(t)$. Let us prove \eqref{9}. \\
Let us consider, \begin{align*}u_{\lambda}(x_i)= E_{x_{i}} e^{-\lambda \tau_{L}} \end{align*}
then,\begin{align*}
\frac{1}{2} \frac{d^{2}u}{d x_i^2}-\lambda u_{\lambda}=0 &&\text{$u_{\lambda}(L_i)=1$, $ i=1,2, \ldots N $ + gluing conditions at 0 given by \eqref{a}, \eqref{b}}\end{align*}
Elementary calculations give \eqref{9}.
\end{proof}




 Using the self similarity property of 1-D Brownian motion $\tau_{L}\overset{\text{law}} {=} L^{2} \tau_{1}$, we will get,
\begin{align}
\psi_{i}(0)=E_{0} e^{-\lambda \tau_{L}}=\frac{1}{\cosh \sqrt{2 \lambda} L} & =E_{0} e^{-\lambda L^{2} \tau_{1}} &&\text{i.e. $\frac{\tau_L}{L^2} = \tau_1$ (law)}
\end{align}
Now we calculate the density for $\tau_{1}$.\\

 Roots of $\cosh \sqrt{2 \lambda}$ are given by the equation,
\begin{align*}
\cosh \sqrt{2 \lambda}=0 \Rightarrow \sqrt{2 \lambda}=i\left(\frac{\pi}{2}+\pi n\right) \Rightarrow \lambda_{n}=-\frac{\pi^{2}(2 n+1)^{2}}{8} && \text{ for $n \geq 0$}
\end{align*}
It leads to the infinite product,
\begin{equation}
\cosh \sqrt{2 \lambda}=\left(1+\frac{8 \lambda}{\pi^{2}}\right) \cdot\left(1+\frac{8 \lambda}{(3 \pi)^{2}}\right) \cdot\left(1+\frac{8 \lambda}{(5 \pi)^{2}}\right) \dots\left(1+\frac{8 \lambda}{((2 n+1) \pi)^{2}}\right) \dots
\end{equation}
Let us find the expansion of Laplace transform of $\frac{1}{\cosh \sqrt{2 \lambda}}$ into simple fractions. It is known that,(~\cite{gradshteyn2014table})
\begin{align*}
\frac{1}{\cos \frac{\pi x}{2}}=\frac{4}{\pi} \sum_{k=1}^{\infty}(-1)^{k+1} \frac{2 k-1}{(2 k-1)^{2}-x^{2}}
\end{align*}
Using the substitution
\begin{align*}
\frac{\pi x}{2}=z \Rightarrow x=\frac{2 z}{\pi}
\end{align*}
and formula
\begin{align*}
\frac{1}{\cosh \sqrt{2 \lambda}}=\frac{1}{\cos i \sqrt{2 \lambda}}
\end{align*}
we will get,
\begin{align*}
\frac{1}{\cosh \sqrt{2 \lambda}} & =\sum_{k=1}^{\infty}(-1)^{k} \frac{4(2 k-1) \pi}{\pi^{2}(2 k-1)^{2}+8 \lambda} \\
& =\sum_{k=1}^{\infty}(-1)^{k+1} \frac{(2 k-1) \frac{\pi}{2}}{\lambda+\frac{\pi^{2}(2 k-1)^{2}}{8}}
\end{align*}
Now applying inverse Laplace transform, we have for the density $P_{\tau_1}(\centerdot)$ of the random variable $\tau_1$ the following fast convergent series
\begin{align}
P_{\tau_1}(s)=\sum_{k=1}^{\infty}(-1)^{k+1} \frac{(2 k-1) \pi}{2} e^{-s \frac{\pi^{2}(2 k-1)^{2}}{8}}
\end{align}
\begin{lemma}
Moments of $\tau_L$ 
are given by the following formulas.
\end{lemma}
\begin{align*}
E_0\tau_{L}^{k} =L^{2k}E\tau_{1}^{k}, && E\tau_{1}^{k}=\frac{E_{2k}}{(2k-1)!}
\end{align*}
here $E_{2k}$ are Euler numbers, see~\cite{gradshteyn2014table}.


Let us prove similar formulas for the complete cycle on $Sp(N, L)$, i.e. the set of N intervals $[0, L]$, connected at point 0 by Kirchhoff's gluing conditions and reflection boundary conditions at the endpoints $L_i=L$, $i=1,2\dots, N$. Such cycles contain the transition from 0 to $\partial L$ and back to 0 from $\partial L$, i.e. ($\tau_0+\tilde{\tau}_0$), then,  
$$E_0 e^{-\lambda(\tau_0+\tilde{\tau}_0)}= \frac{1}{\cosh^{2} \sqrt{2\lambda}L} $$

also, if 
\begin{align*}
T_{N}=\xi_{1}+\xi_{2}+\cdots+\xi_{n} \quad \text { where } \quad \xi_{1}=\left(\tau_{1}+\tilde{\tau_{1}}\right), \dots, \xi_{N}=\left(\tau_{n}+\tilde{\tau_{n}}\right) 
\end{align*}
$\xi_{1}, \xi_{2}, \dots, \xi_{N}$ are $N$ complete Brownian motion cycles on the corresponding spider legs,
then after normalization, 
\begin{align*}
e^{-\lambda \frac{T_{N}}{L^{2}}}=\left(\frac{1}{\cosh ^{2} \sqrt{2 \lambda}}\right)^{N}
\end{align*}
Let us consider now the more general problem of the same type. Instead of $L$-neighborhood of the origin with boundary $\partial_L$ consider the general neighborhood with the endpoints $L_1, \dots, L_N$ on the legs $l_i$, $i=1,2,\dots, N$.\\

\subsection{Exit from the arbitrary neighborhood of the origin}\label{section 5}
 Our goal is to study the distribution of the first exit time $\tau_0(L_1,\dots, L_N)$ and the distribution of the point of the first exit, i.e. point $b\left({\tau_o(L_1,\dots, L_N)}\right)$ for the spider with N legs of lengths $L_1, \dots, L_N$\\

 We start with the second problem. 
Let \begin{align*}
u_i(x)=P_x \left(b({\tau_0}({L_1, \dots, L_N}))=L_i\right) && x \in Sp_N(L_1, \dots, N)
\end{align*}
then (for $i=1$)
\begin{align*}
\frac{1}{2} \frac{d^{2} u}{d x_{i}^{2}}-k u=0 && u_1(L_1)=1, && u_1(L_j)=0 && j \neq1 \quad \text { boundary values + Kirchhoff's gluing condition at $0$ }
\end{align*}
Function $u_i(x)$ is linear on each leg. It is easy to see for $i \neq 1$ 
\begin{align*}
u_i(x_1) = u_i(0) (1- \frac{x}{L_i})&&\text{and}&&u_1(x_1)= \frac{1-u_1(0)}{L_1} x_1 + u_1(0)
\end{align*}
Using Kirchhoff's condition at the origin, we will find 
\begin{align}
u_1(0) = P_0\{b\left( \tau_0(L_1, \dots, L_N)\right) =L_1 \} =  \frac{\frac{1}{L_1}}{\sum_{i=1}^{N}\frac{1}{L_i}}
\end{align}
or in general, 
\begin{align}
u_i(0) = P_0\{ b\left(\tau_0(L_1, \dots, L_N)\right) =L_i \} =  \frac{\frac{1}{L_i}}{\sum_{j=1}^{N}\frac{1}{L_j}}
\end{align}
Of course the maximum exit probability $u_i(0)$ corresponds to the shortest leg: $\min_{j} L_j$ 
We will find the distribution of $\tau_0(L_1, \dots, L_N)$. 
\begin{theorem}
The Laplace transform of the first exit time is
\begin{align*}
E_0 e^{-\lambda(\tau(L_1,\dots, L_n))}=\nu_{\lambda}(0)= \frac{\sum_{i=1}^{N} \frac{\sqrt{2\lambda}}{\sinh \sqrt{2 \lambda} L_{i}}}{\sum_{i=1}^{N}\sqrt{2\lambda}\frac{\cosh \sqrt{2 \lambda} L_{i}}{\sinh \sqrt{2 \lambda} L_{i}}}
\end{align*}
\end{theorem}
\begin{proof}
Let,
$$\nu_{\lambda}=E_x e^{-\lambda \tau_0(L_1, \dots, L_N)} $$
then, 
\begin{align*}
\frac{1}{2} \frac{d^{2} \nu_{\lambda}}{d x^{2}}-\lambda \nu_{\lambda}=0, && \text{with}\qquad \nu_{\lambda}(L_i)=1 && i=1,2,\dots, N
\end{align*}
plus, as usually, the Kirchhoff's condition at 0. 
On each leg $l_i$ 
\begin{align}
\nu_{\lambda}(x_{i})=\frac{\nu_{\lambda}(0) \sinh \sqrt{2 \lambda}\left(L_i-x\right)+\sinh \sqrt{2 \lambda} x_{i}}{\sinh \sqrt{2 \lambda} L_i} 
\end{align}
The unknown parameter $\nu_{\lambda}(0)$ can be found from the gluing condition at 0. 
\begin{align*}
0=\sum_{j=1}^{N} \frac{d \nu_{\lambda}}{d x_{i}} /_{x_{i}=0}=-\nu_{\lambda}(0) \sqrt{2 \lambda} \left(\sum_{i=1}^{N} \frac{\cosh \sqrt{2 \lambda} L_{i}}{\sinh \sqrt{2 \lambda} L_{i}}\right)+\sqrt{2\lambda} \sum_{i=1}^{N} \frac{1}{\sinh \sqrt{2 \lambda} L_{i}}
\end{align*}
then, \begin{align}
\nu_{\lambda}(0)= \frac{\sum_{i=1}^{N} \frac{\sqrt{2\lambda}}{\sinh \sqrt{2 \lambda} L_{i}}}{\sum_{i=1}^{N}\sqrt{2\lambda}\frac{\cosh \sqrt{2 \lambda} L_{i}}{\sinh \sqrt{2 \lambda} L_{i}}}
\end{align}
\end{proof}
\begin{remark} We added the factor $\sqrt{2\lambda}$ in both parts of the fraction for regularization. \end{remark}

Using Taylor's expansion for $\cosh$ and $\sinh$ near $\lambda=0$, one can find 
\begin{align*}
-\frac{d \nu_{\lambda}(0)}{d \lambda}/_{\lambda=0} = E_0 \tau_0(L_1, \dots, L_N) =\frac{\sum_{i=1}^{N} L_i}{\sum_{i=1}^{N}\frac{1}{L_i}}
\end{align*}

If $L_i=L$ then 
$$E_0 \tau_0(L_1\dots, L_N)= E \tau_0(L)=L^2$$

\subsection{Generalization of section \ref{section 5}}

Now we will study different problems of the same type. Let,
\begin{align*}
\tau_{L, 1}=\min \left(t: b_1(t)=L\right)
\end{align*}
without any conditions on the random walk on the legs with numbers $2,3, \dots, N$. The Laplace transform
$u(k, x)=E_{x} e^{-k \tau_{L, 1}}$
satisfies the equation,
\begin{align*}
\frac{1}{2} \frac{d^{2} u}{d x_{1}^{2}}-k u=0, && u(x) /_{(x=x_{1}=L)}=1 \quad+\text { gluing conditions at the origin }+ \text { boundedness of the solutions }
\end{align*}
Then, like in the previous example, we will get
\begin{align*}
u\left(k, x_{j}\right)&=u(0) e^{-\sqrt{2 k} x_{j}}, & j \geq 2 \\
u\left(k, x_{1}\right)&=\frac{u(0) \sinh \sqrt{2 \lambda}\left(L-x_{1}\right)+\sinh \sqrt{2 \lambda} x_{1}}{\sinh \sqrt{2 \lambda} L} &
\end{align*}
From Kirchhoff's equation one can deduce,
\begin{align}
\nonumber
u(\lambda,L) & =E_{0} e^{-\lambda \tau_{L, 1}} \\\nonumber 
& =\frac{1}{\cosh \sqrt{2 \lambda} L+(N-1) \sinh \sqrt{2 \lambda} L} \\
& =\frac{2}{N} \times \frac{e^{-\sqrt{2 \lambda} L}}{1-\left(\frac{N-2}{N}\right) e^{-2 \sqrt{2 \lambda} L}} \label{the2.3}
\end{align}

\begin{theorem}
If $u(\lambda,L)=E_{0} e^{-\lambda \tau_{L, 1}}$, then for $N\geq 2$ 
$$u(\lambda,L) = \frac{2}{N} e^{-\sqrt{2\lambda}L} + \frac{2}{N}(1-\frac{2}{N}) e^{-3\sqrt{2\lambda}L}+ \frac{2}{N}(1-\frac{2}{N})^2 e^{-5\sqrt{2\lambda}L}+\dots$$ 
\end{theorem}

\begin{proof}
 In fact, for $N=2$, the last formula, \eqref{the2.3} gives $\nu_0=e^{-\sqrt{2\lambda}L}$. This is the Laplace transform of the stable law with parameters $\alpha=\frac{1}{2}$, $\beta=1$. The corresponding density has the form 
 
 \begin{align*} P_{L}(s)= \frac{L}{\sqrt{2\pi s^3}} e^{-\frac{L^2}{2s}} \sim \frac{c}{s^{1+\frac{1}{2}}}, & s\to \infty \end{align*}
If, $N>2$ then, 
\begin{align*}
\frac{2}{N} \frac{e^{-\sqrt{2 k} L}}{1-\left(\frac{N-2}{N}\right) e^{-2 \sqrt{2 k} L}}= \frac{2}{N} e^{-\sqrt{2\lambda}L} + \frac{2}{N}(1-\frac{2}{N}) e^{-3\sqrt{2\lambda}L}+ \frac{2}{N}(1-\frac{2}{N})^2 e^{-5\sqrt{2\lambda}L}+\dots
\end{align*}
 This is the geometric progression of the mixture of the stable laws with $\alpha= \frac{1}{2}$, $\beta=1$ and constant $L$, $3L$, $5L$, \dots 
 Note: if we take $\frac{2}{N}=p$, $q= 1-\frac{2}{N}$ then the geometric progression becomes $pe^{-\sqrt{2\lambda}L}+p q e^{-3\sqrt{2\lambda}L}+p q^2 e^{-5\sqrt{2\lambda}L}+\dots $\\
\end{proof}



Consider the new generalization of this model, 
\begin{theorem}
Let, \begin{align*} \tau_{L, N_1}= \min\left(t: x_0(t) =L, \text{for the first $N_1$ legs, $l_1, \dots, l_{N_1}$}\right)
\end{align*} 
without any conditions on the Brownian motion on the remaining $N-N_1$ legs.
\end{theorem}
\begin{proof}
If $\nu (\lambda, L, N_1) = E_0 e^{- \lambda \tau_{L,N_1}}$ then, like in the previous case $N_1=1$, we will get 
\begin{align*}
\nu(\lambda, L, N_1) &=\frac{1}{\cosh \sqrt{2\lambda}L + (N-N_1) \sinh \sqrt{2\lambda}L }\\
&= \frac{2}{N-N_1 +1} \frac{e^{-\sqrt{2\lambda}L}}{\left(1-\frac{2}{N-N_1+1}e^{-2\sqrt{2\lambda}L}\right)}
\end{align*}
For the density of the distribution of $\tau(N,N_1)$ we have almost the same formula
\begin{align*}
P_{\lambda}(s, L, N_1)=\sum_{n=0}^{\infty} p_{N_1} q_{N_1}^{n} \frac{(2 n+1) L}{\sqrt{2 \pi s^3}} e^{- \frac{(2 n+1)^{2}L^2}{2s}}
\end{align*}
\begin{align*}\text{where}&&p_{N_1}= \frac{2}{N-N_1+1},&& q=1-\frac{2}{N-N_1+1} \end{align*}
\end{proof}

\subsection{Structure of the random variable $\tau_{L}$ (moderate and long excursion)}


Now our goal is to understand the structure of the r.v $\tau_{N, L}$: the first exit time from the spider with $N$ legs of the length $L$ ($L>1$ is a large parameter). If the process $b(t)$ starts from $0$ then it will visit each leg infinitely many times at every time interval. To exclude such very short excursions we will divide them into two classes. The first class contains moderate excursions from 0 to $\partial_1$ (i.e., at a distance of 1 from the origin). Such r.v $\tau_{1,i}$ have Laplace transform 
\begin{align*} E_0 e^{-\lambda \tau_{1,i}} = \frac{1}{\cosh \sqrt{2\lambda}}, && \text{where}\qquad E \tau_{1,i} =1 && \text{and} \qquad \mathrm{Var}(\tau_{1,i})= \frac{5}{6} \end{align*}
 After each moderate excursion, we have a potentially long excursion: transition from 1 to 0 or $L$, i.e. the exit time from $[0, L]$. Let us denote such excursion $\tilde{\tau}_{[0,L]}$. 
Then for,  \begin{align*}
E_{x} e^{-\lambda \tau_{[0,L]}}= \nu_{\lambda}(x) \\
\frac{1}{2} \psi_{i}^{\prime \prime}-\lambda \psi_{i}=0\\
\nu_{\lambda}(0)=\nu_{\lambda}(L)=1
\end{align*}
elementary calculations give 
\begin{align*}
\nu_{\lambda}(x) =\frac{\sinh \sqrt{2 \lambda}\left(L-x\right)+\sinh \sqrt{2 \lambda} x}{\sinh \sqrt{2 \lambda} L} 
\end{align*}
i.e. (for the long excursion $x=1$)
\begin{align*}
\nu_{\lambda}(1) =\frac{\sinh \sqrt{2 \lambda}\left(L-1\right)+\sinh \sqrt{2 \lambda}}{\sinh \sqrt{2 \lambda} L} 
\end{align*}
as easy to see, 
$$
 P_{x}\{ b_{\tau_{[0,L]}} =1 \} =\frac{x}{L}$$
The mean length of long plus moderate excursions equals $(L-1)+1=L$, and the process starting from 1 will exit from $Sp(N, L)$ with probability $\frac{1}{L}$. 
We denote the number of such cycles (i.e. excursions) until the first exit from $Sp(N, L)$ by $\nu_{L}$: geometrically distributed r.v. with parameter $\frac{1}{L}$ $\left( P\{ \nu_{L}=k\} = (1-\frac{1}{L})^{k-1}\cdot \frac{1}{L}\right)$. It gives,
\begin{align*}E\nu_L=L, && E \nu_L^{2} =L^2\end{align*} 
\subsection{Limit theorems for the number of cycles}

 By definition, any cycle on the graph $Sp(N; N_1, L)$ containing N legs of length L and $N-N_1$ infinite legs starts from the origin (0) and returns to 0 after visiting one of the points $x_i=1$, $i=1,2, \dots, N$, but not the boundary $\partial_ {L}$. Boundary $\partial_ {L}$ contains $N_1$ points. Consider two cases. \\

\textbf{A)} $N_1=N$, i.e. all $N$ legs have length $L$ (see above) 

 \begin{theorem} \label{theorem2.4} Let $L \to \infty$. Then $P\{ \frac{\nu_L}{L} > x \} =e^{-x}$, $x \geq 0$ \end{theorem}  

\begin{proof}
Is is easy to see that, \begin{align*} E_{0}z^{\nu_{L}} =\frac{z}{L-(L-1)z}, && E\nu_{L}=L \end{align*}
Then, \begin{align*} E_{0}e^{-\frac{\lambda \nu_{L}}{L}} =\frac{e^{-\frac{\lambda}{L}}}{L-(L-1)e^{-\frac{\lambda}{L}}}=\frac{e^{-\frac{\lambda}{L}}}{1+\frac{\lambda(L-1)}{L}+\mathcal{O}(\frac{1}{L^2})}\xrightarrow[L \to \infty]{} \frac{1}{1+\lambda}\end{align*}
and $\frac{1}{1+\lambda}$ is the Laplace transform of $Exp(1)$ law.
\end{proof}

\textbf{B)} $N_1=1$, i.e. only one leg, say $Leg_1$, has length $L$, all other legs are infinite. The creation of the cycles by the Brownian motion on $Sp(N; 1, L)$ includes two Bernoulli experiments: starting from the origin $0$ Brownian motion enters point $x_1=1$ with probability $\frac{1}{N}$ (we call such event, success S) and one of the points $x_j=1$, $j=2, \dots, N$ with probability $\frac{N-1}{N}$ (failure F). The number of failures before the first success, denoted by$\mu_1$ has the geometric law,
 

\begin{align*}  P\{ \mu_1=0\}&= \frac{1}{N}, P\{\mu_1=k\}=\frac{1}{N}\left(\frac{N-1}{N}\right)^k, k>1 \\
 E_{0}z^{\mu_1}& =
 \frac{1}{N-(N-1)z}
\end{align*}

The number of successes S before the first moment, when $b(t)=x_1=1$, has the distribution $\nu_L$ which we already discussed: 
\begin{align*}P\{ \nu_{L}=1\} &= \frac{1}{L}, P\{\nu_L=k \}= \frac{1}{L} (\frac{L-1}{L})^{k-1}, k>1 \\
E_0\nu_L&=\frac{z}{L-(L-1)z}\end{align*}

Finally, the total number of cycles has the form $(\mu_1+1)+ \dots + (\mu_{\nu_L}+1)$ (the r.v $\mu_i$ are independent and have the law of $\mu_1$), i.e.
\begin{align}
 E_{0}z^{\nu_{L,N}}= E_0 z^{\mu_1 + \dots +\mu_{\nu_L}}=\frac{\frac{z}{N-(N-1)z}}{L-(L-1)\frac{z}{N-(N-1)z}}=
 \frac{z}{LN-(LN-1)z}\label{20}\end{align}
Using the previous theorem we will get
\begin{theorem}
For fixed N and $L \to \infty$  \begin{align*}P_0\{ \frac{\nu_{N,L}}{L} >x \} \to e^{-\frac{x}{N}} &&\text{(exponential law with paramter $\frac{1}{N}$, i.e expectation N )} \end{align*} 
\end{theorem}

\begin{proof}
The proof follows from theorem~\ref{theorem2.4} and \eqref{20}.
\end{proof}


\subsection{Time to cover all $N$ legs (Erd\"{o}s-Renyi model)}
We now want to find the asymptotic law for the time $T_N= \{ \text{first moment when Brownian motion $b(t)$ will cover all legs of $Sp(N, L)$} \} $.\\
We will use the following result by Erd\"{o}s-Renyi, related to the Maxwell-Boltzmann experiment~\cite{kolchin1978random}. Let us recall that we impose the boundary condition on the boundary $\partial_{L}$ of $Sp(N, L)$.\\

 Consider N boxes and the random distribution of the particles between boxes. On each step, one particle with probability $\frac{1}{N}$ will be placed into each box. 
We are interested in the random moment $\nu_N$: number of the steps to occupy all N boxes. Of course,\\
$$ \nu_N= \nu_N^{1} + \nu_{N}^{2}+\dots +\nu_N^{N} $$ 
here, \begin{align*}
\nu_{N}^{1} &=1 : \text{one step to occupy one box} \\
\nu_{N}^{2} &= \text{\# of steps to occupy second box, after first step $\nu_N^{1}=1$}\\
\nu_{N}^{3} & =\text{\# of steps to occupy the third box, after the occupation of the second}\\
\dots \\
\nu_{N}^{N} &= \text{\# of steps to occupy the very last empty box, after  $\nu_N^{1}+\dots+ \nu_N^{N-1}$ previous steps}
\end{align*}
The r.v. $\nu_N^{k}$ are geometrically distributed and independent. For $k \geq 2$\\
\begin{align*}
P \{ \nu_N^{k} =n \} = \left(\frac{k-1}{N} \right)^{n-1} \left(\frac{N-k-1}{N} \right) \\
E \nu_{N}=\sum_{k=1}^{N}E\nu_{N}^k=N\left(\sum_{k=0}^{N-1}\frac{1}{N-k}\right)=N\left(\ln N+\gamma+\mathcal{O}(\frac{1}{N})\right),&& \mathrm{Var}\nu_N=\sum_{k=1}^{N}\mathrm{Var}{\nu_N^{k}} =\sum_{k=0}^{N-1}\frac{NK}{(N-k)^2}
\end{align*}
here, $\gamma=.5772$ is Euler's constant. Due to classical result (Erd\"{o}s- Renyi theorem) 
\begin{align} P\{ \frac{\nu_N}{N}- \ln N < x\} \to e^{-e^{-x}} =\mathbb{G}(x)\qquad \text{as $N \to \infty$} \label{nuN}\end{align} 
Then, \begin{align*}P\{ \frac{\nu_N-E\nu_N}{N}<x\}&=P\{\frac{\nu_N}{N}-(\ln N+\gamma+\mathcal{O}(\frac{1}{N}))<x\}\\
&= P\{\frac{\nu_N}{N}-\ln N<x+\gamma+\mathcal{O}(\frac{1}{N})\} \rightarrow e^{-e^{-(x+\gamma)}} =\mathbb{\tilde{G}}(x)\end{align*}

Let us stress that double exponential law has a non-zero mean value: 
$$\int_{\mathbb{R}}x d\mathbb{G}(x)=\int_{\mathbb{R}}x e^{-x} e^{-e^{-x}} dx= \int_{0}^{\infty} \ln{t} e^{-t}dt =\gamma$$
but, $\int_{\mathbb{R}} x d\mathbb{\tilde{G}}=0$,  i.e asymptotically, 
\begin{align*} \nu_N = N \ln N + N \zeta_N && \text{where $ \zeta_N \xrightarrow{\text{law}} \mathbb{G}$ (Gumble distribution)} \end{align*}
Each step in our new experiment is the occupation of one box, that is, complete covering by the trajectory $b(s)$ of one of the legs (transition from 0 to endpoint L of this leg and back to 0, after reflection at point L ). The length of this step is the random variable $\tau_L =L^2 \tilde{\tau}_1$, with 
\begin{align*}
&&Ee^{-\lambda \tilde{\tau_1}} =\frac{1}{\cosh^2\sqrt{2\lambda}} =
1-2\lambda+\frac{8}{3}\lambda^2- \dots \end{align*}
\begin{align*}\text{which implies,}&&a= E \tilde{\tau} = 2
&& E \tilde{\tau}^2 = \frac{8}{3} 
&&\mathrm{var} \tilde{\tau}= \sigma^2 = (2)^2 -\frac{8}{3} =\frac{4}{3}
\end{align*} 
Our goal is to study $T_N= \sum_{j=1}^{\nu_N}\tau_j= L^2 \sum_{j=1}^{\nu_N}\tilde{\tau}_{j}$, that is, total time of occupation by the Brownian motion of $Sp(N, L)$(with reflection condition at all N endpoints L).\\

\begin{theorem}\label{th2.8}
The distribution of the total time for the Brownian motion to cover $Sp(N, L)$, the spider graph with finite number of legs of length $L=1$ has the form \begin{align*}\tilde{T}_N=a N \ln N + a N\zeta_N + \sigma \sqrt{N \ln N } \eta_N + 
+ \mathcal{O}\left(\sqrt{\frac{N}{\ln N }} \right) && \text{where $a=2$ and $\sigma=\sqrt{\frac{4}{3}}$}\end{align*} 
Here $\zeta_{N}$, $\eta_{N}$ are asymptotically independent and have Gumble law \eqref{nuN} and Gaussian law $\mathcal{N}(0,1)$
\end{theorem}

\begin{proof}
 There are two points of view on the limit theorem for $\frac{T_N}{L^2}$. If we know only $T_N$ (but not $\nu_N$)
\begin{align*}\text{Then,}&& 
\frac{T_N}{L^2} =\tilde{T}_N=a \nu_N+(\tilde{T}_N-a \nu_N) 
=a N \ln{N}+a \zeta_N N+\mathcal{O}(N) \end{align*}

that is, \begin{equation}\frac{T_N}{L^2 N \ln N}\xrightarrow[\text{$ N\to \infty $}]{\text{law}} a=2 \end{equation}


 But, if we know the member of $\nu_N$ in our experiment we can use CLT conditionally (for fixed $\nu_N$)
\begin{align*} \sum_{j=1}^{\nu_N}\tilde{\tau}_{j} = \nu_N a + \sigma \sqrt{\nu_N} \eta_{\nu_N} && \text{$\eta_{\nu_N}$ is asymptotically $\mathcal{N}(0,1)$} 
\end{align*}

finally,  $$\tilde{T}_{N}=a N \ln N + a N\zeta_N + \sigma \sqrt{N \ln N } \eta_N + 
+ \mathcal{O}\left(\sqrt{\frac{N}{\ln N }} \right) $$

One can find additional terms of the asymptotic expansion of the distribution of $\frac{T_N}{L^2}=\tilde{T}_{N}$, as $N \to 0$. 
\end{proof}

Let us consider a different version of \ref{th2.8}. Let $Sp(N)$ be the spider graph with N infinite legs and $D(N,L)=\cup_{i=1}^{N}\{0\leq x_i \leq L \} \subset Sp(N)$. Let $\tilde{\tau}_{N, L}$ be the total time for Brownian motion on $Sp(N)$ to cover subset $D(N, L)$. It is now easy to see (comparing with \ref{th2.8})
\begin{align*}
\tilde{\tau}_{N,L}&=\sum_{j=1}^{\nu_N}\tilde{\tau}_j\\
\text{where}\qquad E_0 e^{- \lambda \tilde{\tau}_j}&=\frac{1}{\cosh \sqrt{2\lambda}L} e^{-\sqrt{2\lambda} L}
\end{align*}
(transition from 0 to $L$ for reflected Brownian motion on $[0, \infty)$ and from $L$ to 0). 

\begin{theorem}
Assume that $L=1$, then conditionally for known $\nu_N$ \begin{align*} \frac{\tilde{\tau}_{N,1}}{\nu^2_N}\xrightarrow{\text{law}} \eta_{\frac{1}{2}} \qquad \text{as} \qquad N \to \infty\end{align*}
where the limiting r.v $\eta_{\frac{1}{2}}$ has the positive stable distribution with Laplace transform $exp(-\sqrt{2\lambda})$. Since, $\frac{\nu_N}{(N \ln N)} \xrightarrow{\text{law}} 1$ as $N \to \infty$, we have $\frac{\tilde{\tau}_{N,1}}{(N \ln N)^2}\xrightarrow{\text{law}} \eta_{\frac{1}{2}}$.
\end{theorem}

\begin{proof}
the proof is standard.
\end{proof}

\subsection{Arcsine law and its generalization on the $N$- legged spider graph}
\begin{theorem}\textbf{(P.Levy~\cite{levy1940certains})} Let, b(t), $t \in [0,u]$ be the 1-D Brownian motion and $T=\int_0^{u}\mathbb{I}_{[0, \infty)}(b(t))dt$ is the total time on the positive half axis is given by 
\begin{align}
P\{ t < T < t+dt /b(0)=0\} &= \frac{dt}{\pi \sqrt{t(u-t)}}\label{pl1}\\
P\{T< t /b(0)=0\} &= \frac{2}{\pi} \arcsin\sqrt{\frac{t}{u}}\label{pl2}
\end{align}
\end{theorem}



\begin{proof}
Let us outline the proof of Arcsine law using the Kac-Feynman formula following~(\cite{it1965diffusion}). Consider the 1-D Brownian motion and the random variable $\mathcal{F}(t) =\int_{0}^{t}\mathbb{I}_{[0,\infty)}(b(s)) ds$, that is, the time when the process $b(s)$ spends on positive half-axis $[0, \infty]$. Due to Kac-Feynmann formula, \begin{align*} E_x e^{-\beta \mathcal{F}(t)} dt= u(t,x) \qquad t\geq 0 \qquad x \in \mathbb{R} \end{align*}
is the solution of the equation 
\begin{align}
	\frac{\partial u}{\partial t} &= \frac{1}{2}\mathcal{L}u + \beta  \mathbb{I}_{[0, \infty)}u \qquad \text{where $\mathcal{L}=\frac{1}{2}\frac{\partial^2}{\partial x^{2}}$}\label{23}\\
	\text{with } u(0,x) &= 1 
\label{24}
\end{align}

The corresponding Laplace transform $u_{\beta, \alpha}(x)=\int_{0}^{\infty} e^{-\alpha t} u(t,x)dt$ is the solution of 
\begin{align*}
\frac{1}{2}u^{''}- (\alpha+\beta)u =-1 && x>0\\
\frac{1}{2}u^{''}-\alpha u =-1 && x<0
\end{align*}
Solving these equations and using continuity conditions at $x=0$, we will get (see~\cite{it1965diffusion}), $u_{\alpha, \beta}(0)= \frac{1}{\sqrt{\alpha(\alpha+\beta)}}$.



Let 
$\int_{0}^{t} \mathbb {I}_{[0, \infty)}(b(s)) ds 
= t \theta $ where $\theta= \frac{1}{t}\int_{0}^{t} \mathbb{I}(b(s))ds$,
then by Kac-Feynman's formula, for $\alpha >0$ 
\begin{align*}
u(\alpha, x=0, \beta) &
=\frac{1}{\alpha} \sum_{n=0}^{\infty}(-1)^n m_n (\frac{\beta}{\alpha})^n & \text{ where $m_n=E \theta^n$ and $\frac{\beta}{\alpha}=z$, $n>0$}\\
u(0,z)&=\frac{1}{\sqrt{1+z}}, &\text{ $|z|<1$ generating function of moments}
\end{align*}

One can now prove that, $P\{ \theta <s \} =\frac{2}{\pi} arc \sin{\sqrt{s}}$. In particular, $E\theta_1=\frac{1}{2}$, $E(\tau^{2}_1)=\frac{3}{8}$\\
\end{proof}

Let us find similar formulas for the random variable $\tau_{N,N_0}= \sum_{i=1}^{N_0} \int_{0}^{t} \mathbb {I}_{[0, \infty)}^{i}(b(s)) ds $, i.e., the time Brownian motion $b(s)$ spends on the first $N_0$ legs, $leg_1, \dots ,leg_{N_0}$ of $Sp(N)$.\\
\begin{theorem}
Consider the spider graph with potential $\beta  \mathbb{I}_{[0, \infty)}$ on $N_0$ legs and, no potentials on the $N-N_0$ legs, then the random variable $\tau_{N, N_0}$ follows the generalized arc sine distribution with the generating function of moments $\phi_{\tau_{N, N_0}}=\frac{1}{\sqrt{1+z}} \frac{N_0 + (N-N_0) \sqrt{1+z}}{(N-N_0)+N_0\sqrt{1+z}}$  \end{theorem}
\begin{proof}
Like in the case of Real line $\mathbb{R}$, the Laplace transform for the distribution of $\tau_{N,N_0}$ 
 $$\int_{0}^{\infty} e^{-\alpha t} E_x e^{-\beta \tau_{(N, N_0)}} dt = u_{\alpha, \beta, N, N_0}(x)$$
is the solution of the equation, 
\begin{align*}
u_i^{''}- 2(\alpha+\beta)u_i =-2 && i=1,2, \dots N_0\\
u_{j}^{''}-2\alpha u_j =-2 && j=N_0+1, \dots, N
\end{align*}
Then,
\begin{align*}
u_i =c_i e^{-\sqrt{2(\alpha+\beta)}x_i} +\frac{1}{\alpha+\beta}, && i=1,2, \dots N_0\\
u_j= c_j e^{-\sqrt{2\alpha}x_j} +\frac{1}{2\alpha} && j=N_0+1, \dots, N 
\end{align*} 
plus Kirchhoff's gluing condition at $x=0$. After elementary calculations we get 
\begin{align*} u(0) = 
\frac{1}{\alpha}\frac{1}{\sqrt{1+z}} \frac{N_0+(N-N_0)\sqrt{1+z}}{(N-N_0)+N_0\sqrt{1+z}} && \text{where $\frac{\beta}{\alpha}=z$}
\end{align*}

Then, \begin{align*}\phi_{\tau_{N,N_0}} (z) &= \sum_{n=0}^{\infty}(-)^n z^n m_{n,N,N_0}=\frac{1}{\sqrt{1+z}} \frac{N_0 + (N-N_0) \sqrt{1+z}}{(N-N_0)+N_0\sqrt{1+z}}\\
&=1-\frac{N_0}{N} z + \frac{N_0(N+N_0)}{N^2}\frac{z^2}{2} \dots 
\end{align*}
where $m_{k,N,N_0}=E\theta_{N,N_0}^{n}$, and $\theta_{N,N_0}=\frac{\tau_{N,N_0}}{t}$. This implies, $m_0=1$, the first moment is $m_1= \frac{N_0}{N}$ and the second moment is $m_2=\frac{N_0(N+N_0)}{N^2}$
\end{proof}


\section{The spectral theory of Laplacian on $Sp(N)$ }
In this section, we will develop the direct and inverse Fourier transform on the infinite spider graph. To avoid long formulas we will consider $Sp_3$. As usual, we start from $Sp(3, L)$ with Dirichlet BC at the point $x_i=L$, $i=1,2,3$. Our first goal is to find the most symmetric orthonormal basis of the eigenfunctions for the problem  
\begin{align*} -\frac{1}{2} \frac{d^2 \psi}{d x_i{2}} \psi_{L,i}=\lambda \psi= k^2, \qquad \psi_{i}(L)=0 \end{align*}  
plus standard Kirchhoff's gluing condition and continuity conditions at the origin. There are two different cases: \\








 If $\lambda = k^2 > 0$ and $\psi_{\lambda}(0)=0$ then for $n \geq 1$ there are two eigenfunctions with eigenvalues 
 \begin{equation*}\lambda_n= k_n^2  \Rightarrow k_n= \frac{n \pi}{L} \end{equation*} 
Due to multiplicity 2 of $\lambda_n$, the selection of eigenfunction is not unique. We select the following version: 
\begin{align*}
\psi_{n, 1}(x)= \begin{cases}
0, & x_1 \in  [0,L] \\
\frac{\sin k_{n}x_2}{\sqrt{L}},& x_{2} \in[0, L] \\ 
\frac{-\sin k_{n} x_3}{\sqrt{L}},&  x_{3} \in[0, L] 
\end{cases}
\end{align*}

\begin{align*}
\psi_{n, 2}(x)= \begin{cases}
-\frac{2\sin k_{n}x_1}{\sqrt{3L}},& x_{2} \in[0, L]  \\
\frac{\sin k_{n}x_2}{\sqrt{3L}},& x_{2} \in[0, L] \\ 
\frac{\sin k_{n} x_3}{\sqrt{3L}},&  x_{3} \in[0, L] 
\end{cases}
\end{align*}
as easy to see 

\begin{align*} \int_{Sp(3,L)} \psi_{n,1}^{2} dx = \int_{Sp(3,L)} \psi_{n,2}^{2} dx=1, \qquad \int_{Sp(3,L)} \psi_{n,1} \centerdot \psi_{n,2} dx =0 \end{align*} 
and $\psi_{n,i}$ for $i=1,2$ satisfy condition at $x=0$. Dirichlet BC at $x_i=L$, $i=1,2,3$ is due to the fact, that $k_n$ is the root of the equation $\sin {k_n}L=0$. The third eigenfunction has the form


\begin{align*}
\tilde{\psi}_n(x) =\frac{\cos \tilde{k}_n x_{i}}{\sqrt{3 \frac{L}{2}}}, && i=1,2,3 && \tilde{k}_n = \frac{\pi(n+\frac{1}{2})}{L}, 
\end{align*}
(Neumann condition at $x=0$ implies the Kirchhoff's condition). $\tilde{k_n}$ is very close to $k_n$, this is why we use the notation $\tilde{\psi}_n $ with the same index n, note that $||\tilde{\psi}_n ||=1$ for $n \geq 1$. System of functions $\left(\psi_{n,1}, \psi_{n,2}, \tilde{\psi}_n\right)$ for $n=1,2, \dots$ form orthogonal basis in $L^2(Sp(3,L))$

Consider the compactly supported smooth function $f(x)$ on the $Sp_3$, whose support does not contain the neighborhood of the origin. We will use, in some cases, notations $f(x_i)$, $f_i(x_i)$, $i=1,2,3$ for restrictions of $f(\centerdot)$ on the legs $l_i$, $i=1,2,3$. Assume now that $L$ is sufficiently large and as a result, $\text{support} f(\centerdot) \subset Sp(3, L)$. Let us introduce Fourier transforms on each leg $i$
\begin{align}
\hat{f}_{i, S}(k) = \int_{l_i} f(x_i) \sin kx_i dx_i \label{sin}\\
\hat{f}_{i, C}(k) = \int_{l_i}f(x_i) \cos kx_i dx_ i \label{cos} && i=1,2,3
\end{align}
(indices S, C mean sine, cosine). We can express the Fourier coefficients $a_{n,1}$, $a_{n,2}$, $\tilde{a}_n$ in terms of $\hat{f}_{i,S}$, $\hat{f}_{i,C}$:
\begin{align*}
a_{n,1}&= \int_{0}^{L} f(x_i) \psi_{n,1}(x_i) dx_i= \int_{0}^{\infty} \frac{f(x_2)\sin{k_n}x_2}{\sqrt{L}} dx_2 -  \int_{0}^{\infty} \frac{f(x_3)\sin{k_n}x_3}{\sqrt{L}} dx_3\\
&= \frac{1}{\sqrt{L}} \left[\hat{f}_{2,S}(\frac{\pi n}{L}) - \hat{f}_{3,S}(\frac{\pi n}{L}) \right] \end{align*}

Similarly, 
\begin{align*}
&&a_{n,2}&=\frac{1}{\sqrt{3L}} \left[ -2\hat{f}_{1,S}(\frac{\pi n}{L})+ \hat{f}_{2,S}(\frac{\pi n}{L}) + \hat{f}_{3,S}(\frac{\pi n}{L})\right]\\
\text{and}&&\tilde{a}_{n}&=\frac{1}{\sqrt{\frac{3L}{2}}}\left[ \hat{f}_{1,C} \left(\frac{\pi(n+\frac{1}{2})}{L}\right) + \hat{f}_{2,C}\left(\frac{\pi(n+\frac{1}{2})}{L}\right)+ \hat{f}_{3,C}\left(\frac{\pi(n+\frac{1}{2})}{L}\right) \right]
\end{align*}
The function $f_i(x_i)$, $i=1,2,3$ can be presented by Fourier series 
\begin{align*}
f(x) &= \sum_{n=1}^{\infty} a_{n,1} \psi_{n,1}(x) +  \sum_{n=1}^{\infty} a_{n,2} \psi_{n,2}(x) +  \sum_{n=1}^{\infty} \tilde{a}_{n} \tilde{\psi}_{n}(x) \\
&=\sum_{\rom{1}}+\sum_{\rom{2}}+\sum_{\rom{3}}
\end{align*}
Then, 
\begin{align*}
\sum_{\rom{1}}&=
\sum_{n=1}^{\infty}\left\{
\begin{array}{@{}ll@{}ll@{}}
0, && \text{along leg 1}\\
\frac{1}{\pi}\sin{\pi n x_2}, && \text{along leg 2}\\
-\frac{1}{\pi}\sin{\pi n x_3}, && \text{along leg 3}
\end{array}
\right\}
\left(\hat{f}_{2,S}(\frac{n \pi}{L})-\hat{f}_{3,S}(\frac{n \pi}{L})\right) \frac{\pi}{L}\\
&\overrightarrow{L\to \infty}\frac{1}{\pi}\int_{0}^{\infty}\left\{
\begin{array}{@{}ll@{}ll@{}}
0\\
\sin{k x_2}\\
-\sin{k x_3}
\end{array}
\right\}
\left(\hat{f}_{2,S}(k)-\hat{f}_{3,S}(k)\right) dk
\end{align*}

Similarly, 
\begin{align*}
\sum_{\rom{2}}
\overrightarrow{L\to \infty}\frac{1}{\pi\sqrt{3}}\int_{0}^{\infty}\left\{
\begin{array}{@{}ll@{}ll@{}}
-2 \sin{k x_1}\\
\sin{k x_2}\\
\sin{k x_3}
\end{array}
\right\}
\left(-2 \hat{f}_{1,S}(k)+\hat{f}_{2,S}(k)+\hat{f}_{3,S}(k)\right) dk
\end{align*}

\begin{align*}
\sum_{\rom{3}}
\overrightarrow{L\to \infty}\frac{1}{\pi}\sqrt{\frac{2}{3}}\int_{0}^{\infty}\left\{
\begin{array}{@{}ll@{}ll@{}}
\cos{k x_1}\\
\cos{k x_2}\\
\cos{k x_3}
\end{array}
\right\}
\left(\hat{f}_{1,C}(k)+\hat{f}_{2,C}(k)+\hat{f}_{3,C}(k)\right) dk
\end{align*}

\subsection{Perseval identity on $Sp(N,L)$}
 We start from the equality which expresses the completeness of the eigenfunctions in $L^2\left(sp(3, L)\right)$ 
\begin{align*}
\int_{Sp(3,L)} f^2(x) dx &= \left(\int_{0}^{\infty} f_{1}^{2}(x_1) dx_1 + \int_{0}^{\infty}f_{2}^{2}(x_2) dx_2 + \int_{0}^{\infty}f_{3}^{2}(x_3) dx_3 \right)\\
&= \sum_{n=1}^{\infty} (a_{n,1}^{2}+a_{n,2}^{2}+\tilde{a}_{n}^{2})
\end{align*}
Then, 
\begin{align*}
\sum_{n=1}^{\infty} a_{n,1}^2 &= \sum_{n=1}^{\infty}\frac{1}{L} \left(\hat{f}_{2,S}-\hat{f}_{3,S}\right)^{2} (\frac{\pi n}{L})\rightarrow \int_{0}^{\infty} (\hat{f}_{2,S}-\hat{f}_{3,S})^2 (k) dk
\end{align*}

similarly, 
$$\sum_{n=1}^{\infty} a_{n,2}^2 \rightarrow \frac{1}{3\pi}\int_{0}^{\infty} (-2\hat{f}_{1,S}+\hat{f}_{2,S}+\hat{f}_{3,S})^2 (k) dk $$
$$\sum_{n=1}^{\infty} \tilde{a}_{n}^2 \rightarrow \frac{2}{3\pi}\int_{0}^{\infty} (\hat{f}_{1,C}+\hat{f}_{2,C}+\hat{f}_{3,C})^2 (k) dk $$

Let us describe now from the beginning, the Fourier analysis of $Sp(3)$. If $f(x)$ is a good function then we can calculate on each leg, $l_1$, $l_2$, $l_3$ the cosine and sine Fourier transforms 
\begin{align*}
\hat{f}_{i, S}(k) = \int_{0}^{\infty} f(x_i) \sin kx_i dx_i \\
\hat{f}_{i, C}(k) = \int_{0}^{\infty}f(x_i) \cos kx_i dx_ i && i=1,2,3
\end{align*}
and from their three combinations 
\begin{align*}
\mathcal{F}_{1}(f,k) &= [\hat{f}_{2,S}(k)- \hat{f}_{3,S}(k)]\\
\mathcal{F}_{2}(f,k) &= [-2\hat{f}_{1,S}(k) + \hat{f}_{2,S}(k)+\hat{f}_{3,S}(k)]\\
\mathcal{F}_{3}(f,k) &= [\hat{f}_{1,C}(k) + \hat{f}_{2,C}(k)+\hat{f}_{3,C}(k)]
\end{align*}

 These combinations are the direct Fourier transform of $f$ on $Sp(3)$. Using $\mathcal{F}_{i}(f,k)$, $i=1,2,3$ and the fact that on each leg, $leg_i$, $i=1,2,3$, the standard Fourier transform ($f \to \hat{f}$) is the isometry, we can reconstruct $f(x)$: 
\begin{align*}
f(x)&=\left\{
\begin{array}{@{}ll@{}ll@{}}
f_1(x_1)\\
f_2(x_2)\\
f_3(x_3)
\end{array}
\right\}\\
&= \frac{1}{\pi}\int_{0}^{\infty}\left\{
\begin{array}{@{}ll@{}ll@{}}
0\\
\sin{k x_2}\\
-\sin{k x_3}
\end{array}
\right\}
\mathcal{F}_1(f,k) dk+ \frac{1}{\pi\sqrt{3}}\int_{0}^{\infty}\left\{
\begin{array}{@{}ll@{}ll@{}}
-2 \sin{k x_1}\\
\sin{k x_2}\\
\sin{k x_3}
\end{array}
\right\}
\mathcal{F}_2(f,k) dk \\
&\quad+\frac{1}{\pi} \sqrt{\frac{2}{3}}\int_{0}^{\infty}\left\{
\begin{array}{@{}ll@{}ll@{}}
\cos{k x_1}\\
\cos{k x_2}\\
\cos{k x_3}
\end{array}
\right\}
\mathcal{F}_3(f,k) dk
\end{align*}



\bibliographystyle{amsplain}
\bibliography{references_dbpaper1}

\providecommand{\bysame}{\leavevmode\hbox to3em{\hrulefill}\thinspace}
\providecommand{\MR}{\relax\ifhmode\unskip\space\fi MR }
\providecommand{\MRhref}[2]{%
  \href{http://www.ams.org/mathscinet-getitem?mr=#1}{#2}
}
\providecommand{\href}[2]{#2}
\begin{thebibliography}{1}

\bibitem{berkolaiko2013introduction}
G~Berkolaiko and P~Kuchment, \emph{Introduction to quantum graphs}, no. 186,
  American Mathematical Soc., 2013.

\bibitem{feller1991introduction}
W~Feller, \emph{An introduction to probability theory and its applications,
  volume 2}, vol.~81, John Wiley \& Sons, 1991.

\bibitem{gradshteyn2014table}
I~S Gradshteyn and I~M Ryzhik, \emph{Table of integrals, series, and products},
  Academic press, 2014.

\bibitem{it1965diffusion}
K~Ito and HP~McKean, \emph{Diffusion processes and their sample paths}, Die
  Grundlehren der Mathematischen Wissenschaften in Einzeldarstellungen
  \textbf{125} (1965).

\bibitem{kolchin1978random}
V.F Kolchin, B.A. Sevast'yanov, and V.P. Chistyakov, \emph{Random allocations},
  A Halsted Press book, V. H. Winston, 1978.

\bibitem{levy1940certains}
P~L{\'e}vy, \emph{Sur certains processus stochastiques homog{\`e}nes},
  Compositio mathematica \textbf{7} (1940), 283--339.

\bibitem{molchanov2007scattering}
S~Molchanov and B~Vainberg, \emph{Scattering solutions in networks of thin
  fibers: small diameter asymptotics}, Communications in mathematical physics
  \textbf{273} (2007), no.~2, 533--559.

\bibitem{paul2023spectral}
M~Paul, \emph{Spectral theory of schr{\"o}dinger type operator on spider type
  quantum graphs}, Ph.D. thesis, The University of North Carolina at Charlotte,
  2023.

\end{thebibliography}











\end{document}